\documentclass[12pt,leqno]{amsart}

\usepackage{latexsym,amsmath,amssymb,amsthm,amsfonts}

\usepackage{geometry}                
\usepackage{graphicx}
\usepackage{epstopdf}
\renewcommand{\epsilon}{\varepsilon}

\newcommand{\R}{\mathbb{R}}

\DeclareMathOperator{\diam}{diam}

\newtheorem{theorem}{Theorem}
\newtheorem{lemma}[theorem]{Lemma}
\newtheorem{proposition}[theorem]{Proposition}

 \title[Mesh ratios for best-packing]{Mesh ratios for best-packing and limits of   minimal energy configurations}

\author{ A. V. Bondarenko$^*$, D. P. Hardin$^\dagger$, and E. B. Saff$^\dagger$} 
\thanks{ $^*$The research of this author was conducted while visiting the Center for Constructive Approximation in the Department of Mathematics,
Vanderbilt University and the Mathematisches Forschungsinstitut Oberwolfach.\\
 $^\dagger$The research of these authors  was supported, in part, by the U. S. National Science Foundation under grant
  DMS-1109266. }

\date{\today}

\address{A.  V. Bondarenko: Centre de Recerca Matem\`atica Campus de Bellaterra, Edifici C, 08193
Bellaterra (Barcelona), Spain and Department of Mathematical
Analysis, National Taras Shevchenko University, str. Volodymyrska,
64 Kyiv, 01033, Ukraine\\ D. P. Hardin and E. B. Saff: Center for
Constructive Approximation, Department of Mathematics, Vanderbilt
University, Nashville, TN 37240, USA } \email{andriybond@gmail.com}
\email{Doug.Hardin@Vanderbilt.Edu}
\email{Edward.B.Saff@Vanderbilt.Edu}


\begin{document}

\begin{abstract}
For $N$-point best-packing configurations $\omega_N$ on a compact 
metric space $(A,\rho)$, we obtain estimates for the mesh-separation ratio $\gamma(\omega_N,A)$, which is the quotient of the covering radius of $\omega_N$ relative to $A$ and the minimum pairwise distance between points in $\omega_N$.  For best-packing configurations $\omega_N$ that arise as limits of
minimal Riesz $s$-energy   configurations as $s\to \infty$, we prove that $\gamma(\omega_N,A)\le 1$ and this bound can be attained even for the sphere.  In the particular case when $N=5$ on $S^2$ with $\rho$ the Euclidean metric, we prove our main result that among the infinitely many 5-point best-packing configurations there is a unique  configuration, namely a square-base pyramid $\omega_5^*$, that is the  limit (as $s\to \infty$) of 5-point $s$-energy minimizing configurations.  Moreover,   $\gamma(\omega_5^*,S^2)=1$.
\end{abstract}

\maketitle

\section{Introduction}
Let $A$ be a compact infinite metric space with metric $\rho: A\times A\rightarrow[0,\infty)$ and  let $\omega_N=\{x_i\}_{i=1}^N\subset  A$ denote a configuration of $N\ge 2$ points in $A$.
We are chiefly concerned with  two `quality' measures of $\omega_N$; namely, the {\em separation distance of $\omega_N$}   defined by
\begin{equation}\delta(\omega_N)=\delta^\rho(\omega_N):=\min_{1\leq i\not= j\leq N}\rho(x_i,x_j),
\end{equation}
and    the   {\em mesh norm}  {\em of $\omega_N$ with respect to   $A$}   defined by
 \begin{equation}\eta(\omega_N,A)=\eta^\rho(\omega_N,A):=\max_{y\in A}\min_{1\leq i\leq N}\rho(y,x_i).
 \end{equation}
 This quantity is also known as the  {\em fill radius} or {\em covering radius} of $\omega_N$ relative to $A$.
The optimal values of these quantities are also of interest and we consider, for $N\ge 2$,  the {\em $N$-point best-packing
 distance on $A$} given by 
 \begin{equation}\label{deltaNdef}
 \delta_N(A)= \delta_N^\rho(A):=\max \{\delta(\omega_N)\colon \omega_N\subset  A,\, \#\omega_N =N\},
 \end{equation}
 and the {\em $N$-point mesh norm}  of  $A$  given by
 \begin{equation}\label{etaNdef}
\eta_N(A)=\eta_N^\rho(A):=\min  \{\eta(\omega_N,A)\colon \omega_N\subset  A,\, \#\omega_N=N\},
\end{equation}
 where $\#S$ denotes the cardinality of set $S$.   A configuration  $\omega_N$ of $N$ points in $A$ is called a {\em best-packing configuration for $A$} if 
 $\delta(\omega_N)=\delta_N(A)$.

  In the theory of approximation and interpolation (for example, by splines or radial basis functions (RBFs)), the separation distance is often associated with some measure of `stability' of the approximation, while the mesh norm arises in the  error of the approximation. In this context, the  {\em mesh-separation ratio} (or {\em mesh ratio}) $$\gamma(\omega_N,A):=\eta(\omega_N,A)/\delta(\omega_N),$$
 can be regarded as a `condition number' for $\omega_N$ relative to $A$.   If $\{\omega_N\}_{N=2}^\infty$ is a sequence of $N$-point configurations such that $\gamma(\omega_N,A)$ is uniformly bounded in $N$, then the sequence is said to be {\em quasi-uniform on $A$}.
    Quasi-uniform sequences of configurations are important for a number of methods involving RBF approximation and interpolation (see \cite{ FW, LGSW, P, S}).

 We remark that in some cases it is easy to obtain positive lower bounds for the mesh-separation ratio.  For example, if $A$ is connected, then $\gamma(\omega_N,A)\ge 1/2$.   Furthermore, letting
 $$B(x,r):=\{y\in A: \rho(y,x)\leq r\}$$ denote the closed ball in $A$ with center $x$ and radius $r$,
  then
 $\gamma(\omega_N,A)\ge \beta/2$  for any $N$-point configuration $\omega_N\subset  A$ whenever $A$ and $\beta\in(0,1)$ have the property that for any $r\in (0,\diam (A)]$ and any $x\in A$, the annulus $B(x,r)\setminus B(x ,\beta r)$
 is nonempty.     The {\it diameter} of $A$ is defined by $$\diam(A):=\max\{\rho(x,y)\colon x\in A,\ y\in A\}.$$
 
 The outline of the paper is as follows.  In Section~2 we present two simple but basic results concerning  the mesh-separation ratio for best-packing configurations on general sets.  In Section~3, we obtain lower bounds for this ratio for any best-packing configuration on the sphere in $\R^n$ and, in Section~4, we study the special case of minimal Riesz $s$-energy 5-point configurations on $S^2$ and determine their limiting best-packing configuration as $s\to \infty$.  Section~5 is devoted to a brief discussion of some special best-packing configurations on $S^n$.
 
\section{Mesh-separation ratio for general sets}
The following simple result is of the same 
spirit as that of Proposition 2.1 of \cite{Wardetal}.

\begin{theorem}\label{Thm0}
Let $(A,\rho)$ be a compact infinite metric space.  Then, for each $N\ge 2$, there exists an $N$-point best-packing configuration $\omega_N$    on $A$ such that $\gamma(\omega_N,A)\le 1$.  In particular, this holds for any best-packing configuration $\omega_N=\{x_1,\ldots, x_N\}$ having the minimal
number of  pairs of points $\{x_i, x_j\}$ such that
$\rho(x_i,x_j)=\delta_N(A)$.
\end{theorem}
\begin{proof}
Let $\omega_N$ be a best-packing configuration on $A$ having the minimal
number of unordered pairs of points $\{x_i, x_j\}$ such that
$\rho(x_i,x_j)=\delta_N(A)$. If $\eta(\omega_N,A)>\delta_N(A)$,
then select a point $a\in A$ such that
$\rho(a,x_i)>\delta_N(A)$ for $i=1,\ldots,N$, and  choose a point $x_\ell$ from some pair $\{x_k,x_\ell\}$ such that $\rho(x_k,x_\ell)=\delta_N(A)$.  Let $\omega_N'$ be the best-packing configuration obtained by replacing $a$ in $\omega_N$ by $x_\ell$.    Clearly, $\omega_N'$ has fewer
unordered pairs of points $\{x_i, x_j\}$ such that
$\rho(x_i,x_j)=\delta_N(A)$ than $\omega_N$. This contradiction proves
Theorem~\ref{Thm0}.
\end{proof}
On the other hand, there exist examples of compact metric spaces
$(A, \rho)$ for which
\begin{equation}
\label{supsup}
\limsup_{N\to\infty}\sup \{\gamma(\omega_N,A)\mid \delta(\omega_N)=\delta_N(A)\}=\infty,
\end{equation}
  as we now show.\\
{\bf Example 1.} Let $A$ be the standard $1/3$ Cantor set in [0,1]  and let $\rho$ be the
Euclidean metric.  
For each $N\in{\mathbb N}$, the set $A$ is  contained in the union of $2^N$
disjoint intervals of length $3^{-N}$ with endpoints $0=x_1^N<x^N_2<\ldots<x^N_{2^{N+1}}=1$ which belong to $A$.  For any configuration of $2^N+1$ points in $A$,  at least one of the intervals  of length $3^{-N}$ must contain at least two points from the configuration showing that $\delta_{2^N+1}(A)\le
3^{-N}$. On the other hand, the   configuration 
$\omega_{2^N+1}:=\{x_1^N,\ldots, x_{2^N+1}^N=2/3\}$
is a best-packing configuration since $\delta(\omega_{2^N+1})=\delta_{2^N+1}(A)= 3^{-N}$
and has mesh norm $\eta(\omega_{2^N+1},A)=1/3$.  Thus \eqref{supsup} holds. \vspace{5mm}

Best-packing configurations arise as limits of minimum energy configurations as we now describe. 
For  a configuration $\omega_N:=\{x_1,\ldots ,x_N\}\subset A$  of $N\geq 2$ distinct points
and $s>0$,
 the  {\em   Riesz $s$-energy of $\omega_N$} is defined by
$$
E_s(\omega_N)=E_s^\rho(\omega_N):=\sum_{1\leq i\neq j\leq N}{\frac
{1}{\rho(x_i,x_j)^s}}=\sum_{i=1}^N\sum_{ {j=1} \atop { j\neq i}}^N\frac{1}{\rho(x_i,x_j)^s},
$$
while  the {\em $N$-point Riesz $s$-energy of $A$} is defined by
\begin{equation} \label{c2'}\mathcal E_s(A,N)=\mathcal E_s^\rho(A,N):=\inf \{E_s(\omega_N) : \omega_N\subset A ,
 \# \omega_N =N\}.
\end{equation}    An $N$-point configuration $\omega_N\subset A$ is said to {\it $s$-energy minimizing} if 
$E_s(\omega_N)=\mathcal E_s(A,N)$.

\begin{proposition}[\cite{BHS07}]\label{prop2}
Let $(A,\rho)$ be an infinite compact metric space.  For each fixed $N\ge 2$, 
$$
\lim_{s\to \infty}\mathcal{E}_s(A,N)^{1/s}=\frac{1}{\delta_N(A)}.
$$
Moreover, every cluster point as $s\to \infty$ of $s$-energy minimizing $N$-point configurations on $A$  is an $N$-point best-packing configuration on $A$.
\end{proposition}

The following theorem concerning the mesh-separation ratio of  best-packing
configurations that arise as cluster points of $s$-energy minimizing
configurations generalizes, simplifies, and  improves Theorem~7 of~\cite{HSW}.
\begin{theorem}\label{Thm0.5}
For a fixed $N\ge 2$, let $\omega_N$ be a cluster point as $s\to\infty$ of a family of
$N$-point $s$-energy minimizing configurations on a compact metric
space $(A,\rho)$. Then
$\gamma(\omega_N,A)\le1$.
\end{theorem}
The upper bound for $\gamma(\omega_N,A)$ in this theorem can be attained even for the case
when $A$ is a sphere and $\rho$ is the Euclidean metric.  For $N=11$ on $S^2$, equality follows from the uniqueness result for best-packing of B\"or\"oczky   \cite{B2}.  For $N=5$ on $S^2$, it follows from Theorem~\ref{Q_4} in Section~4.  
\begin{proof} Let $N\ge 2$ be fixed and, for $s>0$, let $\omega_{N,s}$
be an $N$-point $s$-energy minimizing configuration 
on $A$. Clearly, $E_s(\omega_{N,s})\ge \delta(\omega_{N,s})^{-s}$. This
implies that there exists a point $x_s\in \omega_{N,s}$ such that
$$
\sum_{y\in \omega_{N,s}\setminus\{x_s\}}\rho(x_s,y)^{-s}\ge N^{-1}E_s(\omega_{N,s})\ge N^{-1}\delta(\omega_{N,s})^{-s}.
$$
If
$ \eta(\omega_{N,s},A)> N^{2/s}\delta(\omega_{N,s}),$
then $E_s(\omega_{N,s}')< E_s(\omega_{N,s})$, where
$\omega_{N,s}':=\omega_{N,s}\cup\{a\}\setminus\{x_s\}$, and $a$ is a point of $A$ such
that $\rho(t,a)\ge\eta(\omega_{N,s},A)$, for all $t\in \omega_{N,s}$, which yields a contradiction.  
Hence,
\begin{equation}
\label{mesh} \eta(\omega_{N,s},A)\le N^{2/s}\delta(\omega_{N,s}),
\end{equation}
and letting $s\to\infty$
in~\eqref{mesh} and using Proposition~\ref{prop2}, we obtain the statement of Theorem~\ref{Thm0.5}.
\end{proof}
\section{Lower bounds for the mesh-separation ratio on the sphere}

In this section we derive some lower bounds for the mesh-separation ratio of a best-packing $N$-point configuration on the unit sphere $S^n\subset \R^{n+1}$ with $\rho$ the Euclidean metric.
Let $\Delta_n$ and $\Theta_n$ be the sphere packing and covering
constants in $\R^n$, respectively:
\begin{equation}\label{Delta}
\Delta_n:=\lim_{N\to \infty}N({\delta_N(U_n) }/{2})^n\beta_n \ \ \text{ and }\ \  \Theta_n:=\lim_{N\to \infty} N\eta_N(U_n)^n \beta_n,
\end{equation}
where $U_n:=[0,1]^n$ denotes the unit cube in $\R^n$ and  $\beta_n$ denotes the volume of the unit ball in $\R^n$   
(see, e.g.~\cite{CS,K}).
First we  prove the following asymptotic result for best-packing configurations on  $S^n$.
\begin{theorem}\label{Thm0.75}
Let $\{\omega_N\}$ denote a sequence of $N$-point  best-packing configurations on $S^n$.
Then
\begin{equation}
\label{asymp}
\gamma(\omega_N,S^n)\ge\frac12\left(\frac{\Theta_n}{\Delta_n}\right)^{1/n}+o(1),\quad
N\to\infty.
\end{equation}
\end{theorem}

\begin{proof}
 Since
the collection of spherical caps with centers in the points of $\omega_N$
of the radius $\eta(\omega_N,S^n)$ covers $S^n$, a standard
projection argument  implies
\begin{equation}
\label{cov}
N\beta_n\left(\eta(\omega_N,S^n)\right)^n\ge\Theta_n {\rm Area}(S^n)+o(1),\quad
N\to\infty.
\end{equation}
Similarly we have
\begin{equation}
\label{pack}
N\beta_n\left(\frac{\delta(\omega_N)}2\right)^n\le\Delta_n{\rm Area}(S^n)+o(1),\quad
N\to\infty.
\end{equation}
Thus, we obtain~\eqref{asymp} directly from~\eqref{cov}
and~\eqref{pack}.
\end{proof}

It is interesting to investigate the asymptotic behavior of the
constant on the right-hand side of~\eqref{asymp}  as
$n\to\infty$. The best known asymptotic upper bound for
$\Delta_n$ is the Kabatyanski-Levenshtein bound $\Delta_n\le
2^{-0.599n+o(n)}$ as $n\to\infty$ and the best known lower bound for the
covering constant is $\Theta_n\ge cn$, where $c$ is a positive
absolute constant (cf. \cite[pages 40 and 247]{CS}). Thus the inequality~\eqref{asymp} implies the
following: if $n$ is large enough and $N>C(n)$, then the inequality
$$
\gamma(\omega_N,S^n)\ge (1/2)2^{0.599}+o(1)\ge 0.757, \qquad n\to \infty,
$$
holds for an arbitrary best-packing configuration $\omega_N$ on
$S^n$. Further upper bounds for $\Delta_n$ and lower
bounds for $\Theta_n$ can be found in~\cite{CE} and~\cite{CS}. In particular, it is known that for $n=2$ the hexagonal
lattice provides both $\Delta_2=\pi/\sqrt{12}$, and
$\Theta_2=2\pi/\sqrt{27}$. Hence
$$
\gamma(\omega_N,S^2)\ge\frac{1}{\sqrt{3}}+o(1),\quad N\to\infty,
$$
for an arbitrary best packing configuration $\omega_N$ on $S^2$.
However, by special arguments working only for $n=2$ we are able to
improve this result to the following:
\begin{theorem}\label{Thm0.8}
Let $\{\omega_N\}$ denote a sequence of $N$-point  best-packing configurations on $S^2$.
Then
\begin{equation}
\label{asymp2} \gamma(\omega_N,S^2)\ge\frac
1{2\cos{\pi/5}}+o(1)=\frac{2}{1+\sqrt{5}}+o(1),\quad
N\to\infty.
\end{equation}
\end{theorem}
\begin{proof} It suffices to only consider sequences such that $\gamma(\omega_N,S^2)=O(1)$ as $N\to \infty$. 
For  a fixed $N\ge 4$,
consider the Voronoi decomposition of $S^2$ generated by $\omega_N$, with $X_i$ denoting the cell
associated with $x_i$; that is, 
$$
X_i:=\{v\in S^2\, | \,|v-x_i|=\min_{x\in\omega_N}|v-x|\}.
$$
  Euler's formula for convex polyhedra implies that there is a cell $X_j$ having at
most 5 edges (each cell is a spherical polygon with edges consisting
of arcs of great circles), see~\cite{HS}.  
   Since  $$B(x_i,\delta(\omega_N)/2)\cap S^2\subset X_i\subset B(x_i,\eta(\omega_N,S^2)),$$ and 
   $\eta(\omega_N,S^2)=O(\delta(\omega_N))$, it follows by a projection argument 
  that    there is at least one interior angle from $x_j$ to consecutive vertices of $X_j$ with angle $2\pi/5+o(1)$, 
  and hence the distance from $x_j$ to some  vertex of $X_j$ is
at least $$
\frac{\delta(\omega_N)}{2\cos{\pi/5}}+o(\delta(\omega_N)),\quad N\to\infty.
$$
This yields~\eqref{asymp2}.

\end{proof}

\section{Limit of minimal energy for 5 points on $S^2$}

\begin{figure}[htbp] 
    \centering
    \includegraphics[width=1.9in]{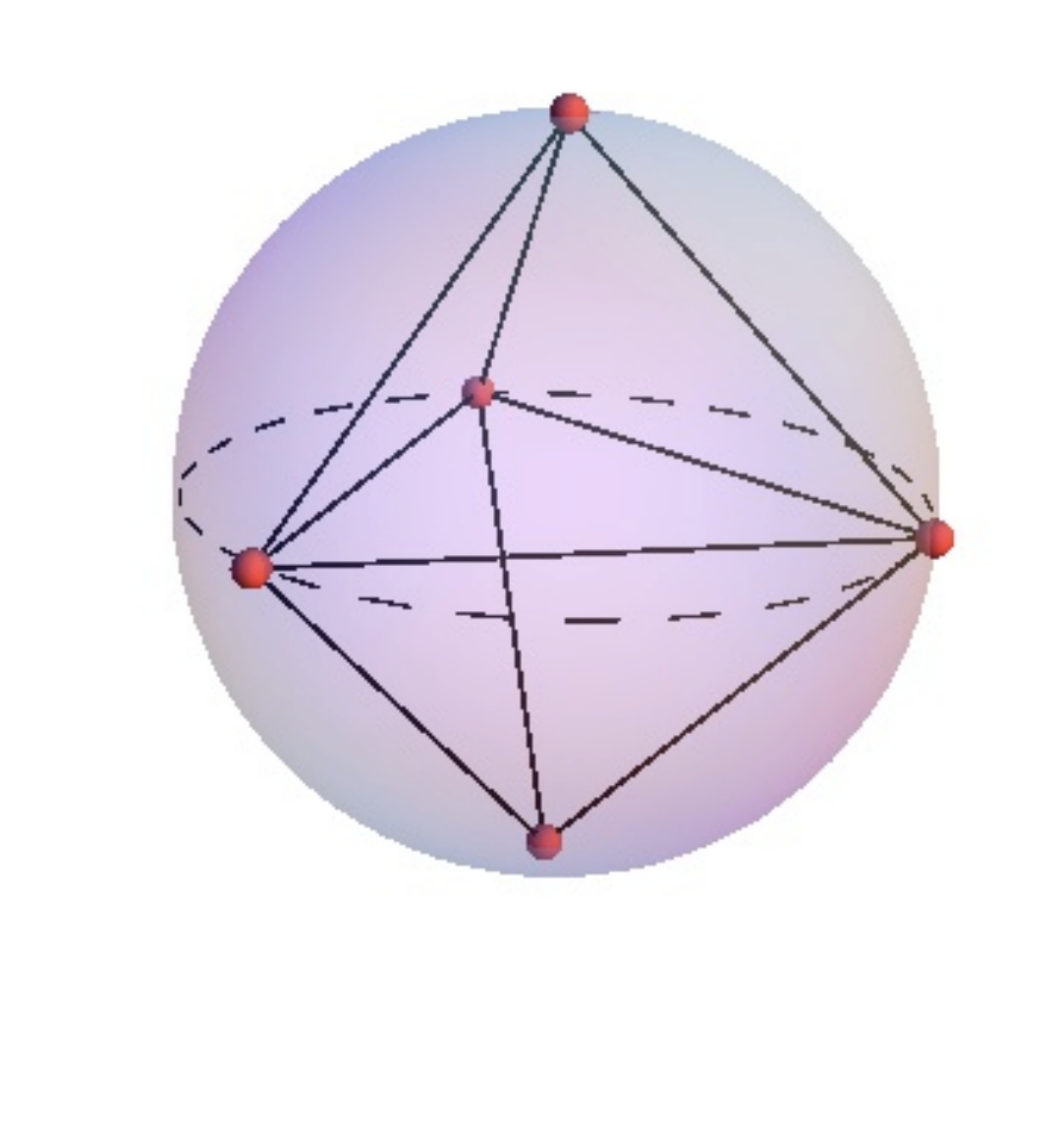}  \includegraphics[width=1.9in]{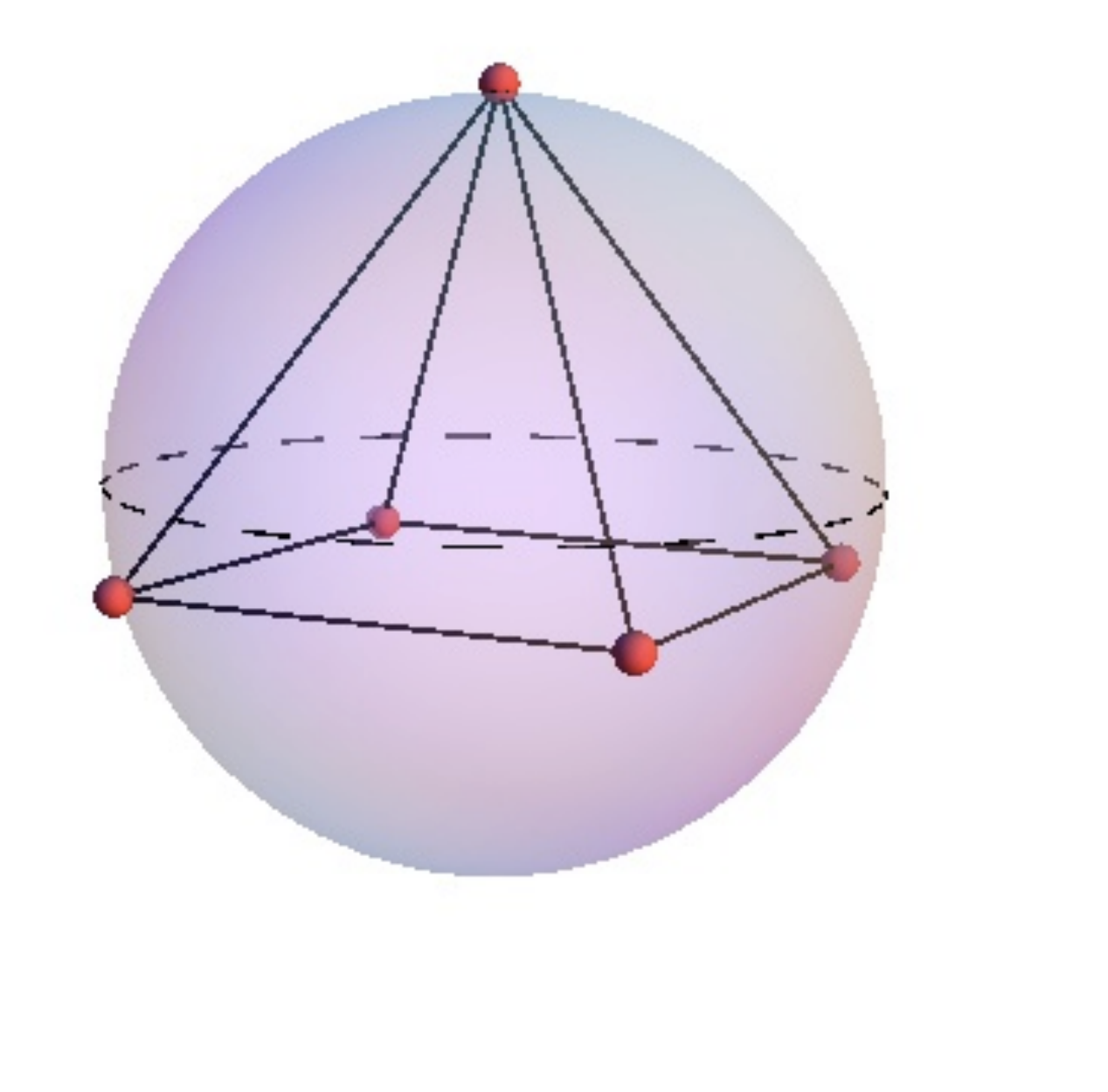} \includegraphics[width=1.9in]{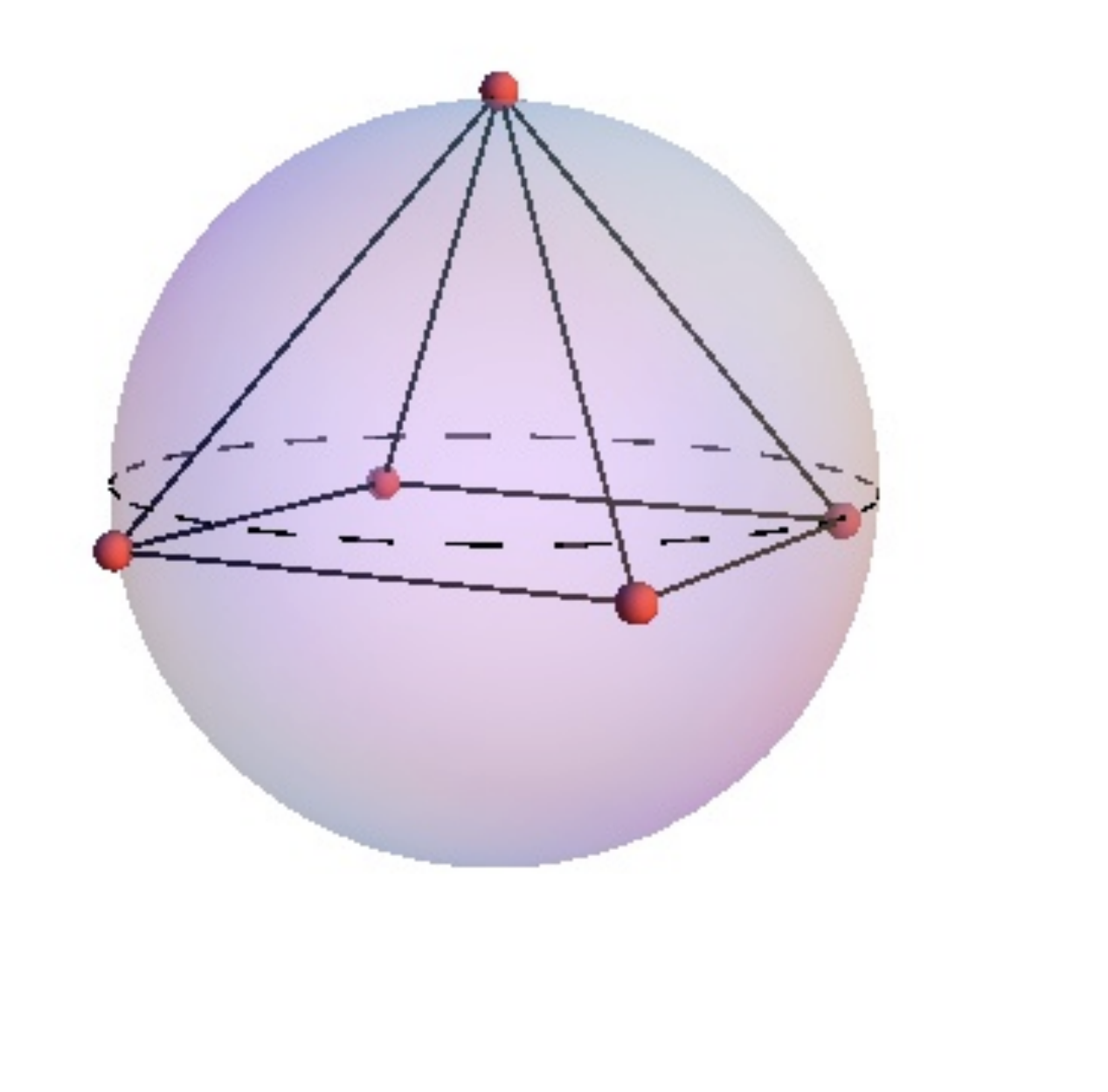}
    
    \vspace{-.1in}
    \hspace*{-.1in} (a)   \hspace{1.3in} (b)   \hspace{1.6in} (c)  
       \caption{`Optimal' 5-point configurations on ${\mathbb S}^2$: (a)  bipyramid BP,  (b) optimal square-base pyramid  SBP(1) , (c) optimal square-base pyramid   SBP(16).}
       \label{fivepointFigs}
 \end{figure}
 
It was observed in \cite{MKS} from numerical experiments that 5-point minimum Riesz $s$-energy configurations on $S^2$ with the Euclidean metric appear to depend on $s$ and  to be of two general types:   (i) the bipyramid (BP) consisting of 2 antipodal points and 3 equally spaced points on the associated equator, and   (ii) the square-base pyramid (SBP$(s)$) with one vertex at the north pole and 4 vertices of the same latitude depending on $s$ and forming a square (see Figure~\ref{fivepointFigs}).   
A comparison of the $s$-energy for the BP and the SBP$(s)$ configurations is given in Figure~\ref{fivepointratioFig} and suggests (as in \cite{MKS}) that   BP is optimal for $s<s^*\approx 15.04808$, while SBP$(s)$ is optimal for $s>s^*$.

R. Schwartz \cite {Schw} using a mathematically rigorous computer-aided solution proved (in a manuscript of 67 pages)  that, for $N=5$, BP is the unique  minimizer of the Riesz $s$-energy for $s=1$ and $s=2$.  (For the logarithmic energy, the optimality of BP is established in \cite{DLT}.)  Currently there are no other values of $s>0$ for which a rigorous optimality proof is known.   Regarding the stability  of BP and SBP($s$), in Figure~\ref{fivepointstabilityFig} we plot the minimum eigenvalue of the Hessian of their $s$-energies.  These graphs suggest that BP is not a local minimizing configuration for $s>21.148$ (also observed by H. Cohn), while SBP($s$) is not a  local minimizing configuration for $s<13.5204$.

\begin{figure}[htbp] 
    \centering
    \includegraphics[width=5in]{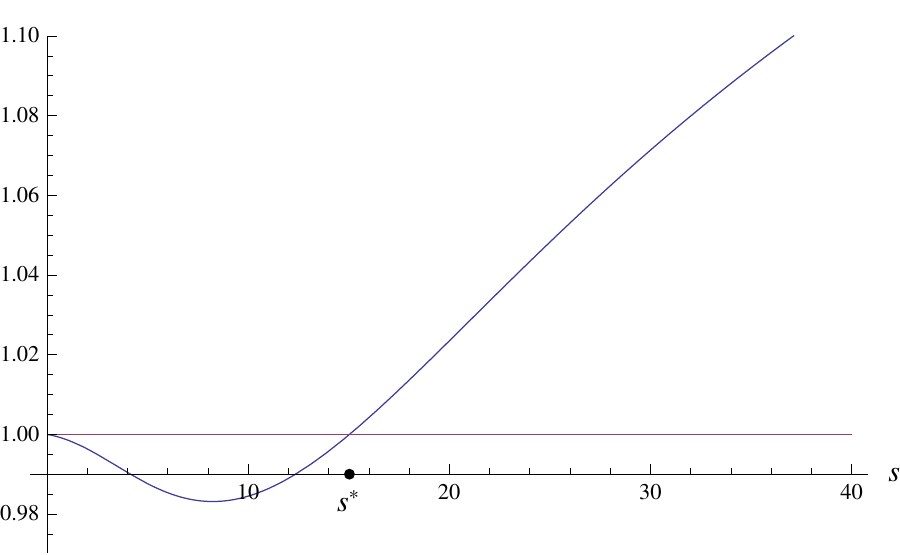}             \caption{The ratio of the $s$-energy of BP to the $s$-energy of SBP($s$). }
       \label{fivepointratioFig}
 \end{figure}
 
 \begin{figure}[htbp] 
    \centering
    \includegraphics[width=2.5in]{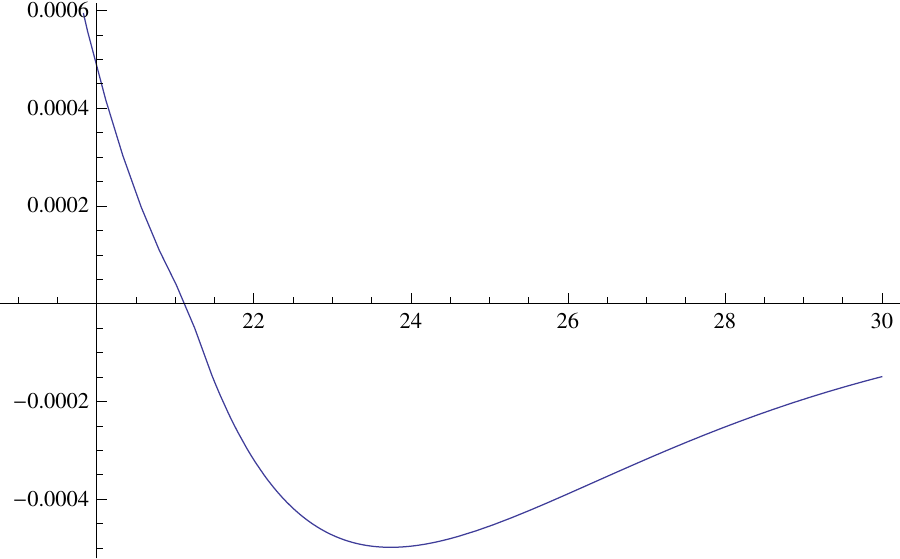}     \includegraphics[width=2.5in]{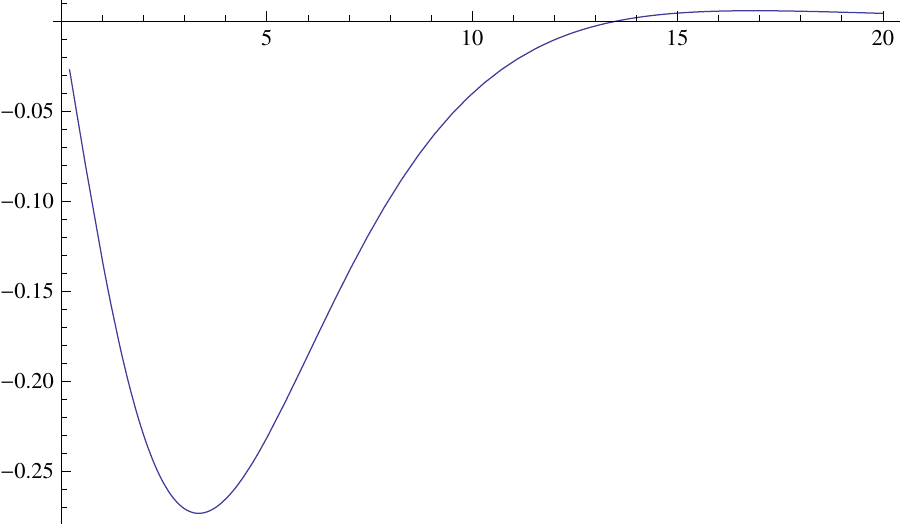}      
    
    (a) \hspace{2.5in} (b)
          \caption{The minimum eigenvalue of the Hessian for the $s$-energy of (a) the BP configuration and (b) the SBP($s$) configuration. }
       \label{fivepointstabilityFig}
 \end{figure}

According to Proposition~\ref{prop2}, every cluster point of $s$-energy minimizing configurations as $s\to \infty$ is a best-packing configuration.  However, as is known, there are infinitely many non-isometric 5-point best-packing configurations on $S^2$ (see e.g. \cite{KB}).   
\begin{proposition}
\label{l1} $\delta_5(S^2)=\sqrt{2}$ and all 5-point best-packing configurations on $S^2$ consist
of two antipodal points (poles) and a triangle on the equator having
all angles greater than or equal to $\pi/4$.
\end{proposition}

It appears from Figure~\ref{fivepointratioFig} that the unique (up to isometry) cluster point of 5-point $s$-energy minimizing configurations is  SBP($\infty$); that is, the square base pyramid with base on the equator.   We shall next provide a rigorous proof that this is indeed the case.

\begin {theorem}\label {Q_4}
Let $Q'$ be a cluster point of a family of $5$-point $s$-energy minimizing configurations on $S^2$ as $s\to \infty$. Then $Q'$ is isometric to
\begin {equation}\label {Qp}
Q={\rm SBP}(\infty):=\{{\rm e}_1,-{\rm e}_1,{\rm e}_2,{\rm e}_3,-{\rm e}_3\},
\end {equation}
where ${\rm e}_1=(1,0,0)$, ${\rm e}_2=(0,1,0)$, and ${\rm e}_3=(0,0,1)$.
\end {theorem}

 It is perhaps surprising that this configuration has the maximum number of common pairwise distances (eight) of length $\sqrt{2}$ among all $5$-point best-packings.

We start the proof with an upper estimate for the minimum 5-point $s$-energy on $S^2$.  
\begin {lemma}\label {S_8}
$$
\limsup\limits_{s\to\infty}{2^{s/2}\mathcal E_s(S^2,5)}\leq 8.
$$
\end {lemma}
\begin {proof}
For arbitrary $0<t<1$, we define the following 5-point configuration on $S^2$:
\begin {equation}\label {2Q}
Q_t:=\{(\pm \sqrt {1-t^2},-t,0),(0,-t,\pm\sqrt {1-t^2}),{\rm e}_2\},
\end {equation}
which, for a suitable choice of $t$ (depending on $s$), is a conjectured  minimum energy configuration on $S^2$ for every $s$ large enough.
The $s$-energy of this configuration is given by
$$
E_s(Q_t):=4\cdot 2^{-s}(1-t^2)^{-s/2}+8\cdot 2^{-s/2}(1-t^2)^{-s/2}+8\cdot 2^{-s/2}(1+t)^{-s/2}.
$$
Letting now $t=s^{-2/3}$, we obtain that 
$$
\lim_{s\to \infty}(1-t^2)^{-s/2}=1\ \ \ {\rm and}\ \ \ \lim_{s\to \infty}{(1+t)^{-s/2}}=0,
$$
and so
$$
\limsup\limits_{s\to\infty}{2^{s/2}\mathcal E_s(S^2,5)}\leq \lim\limits_{s\to \infty}{2^{s/2}E_s(Q_{t})}
$$
$$
=\lim\limits_{s\to \infty} (4\cdot 2^{-s/2}(1-t^2)^{-s/2}+8(1-t^2)^{-s/2}+8(1+t)^{-s/2})=8.
$$
\end {proof}
We further need  the following  statement.
\begin {lemma}\label {M}
Let $A,B$, and $M$ be fixed positive constants. Then
$$
f(x):=M(1-Ax)^{-s}+(1+Bx)^{-s}\geq M+\min \{1,AM/B\}
$$
for every $x\in [0,1/A)$ and $s>0$.
\end {lemma}
\begin {proof}
It is not difficult to see that $f$ attains its minimum on $[0,1/A)$ at the point $x_0=0$ if $B\leq AM$ and at the point 
$$
x_1=\frac {(B/(AM))^{1/(s+1)}-1}{B+A(B/(AM))^{1/(s+1)}}
$$
if $B>AM$. In the first case we have
$$
f(x)\geq f(0)=M+1, \ \ x\in [0,1/A),\ \ s>0.
$$
In the second case, since
$$
x_1\leq \frac {1}{B}\left[(B/(AM))^{1/(s+1)}-1\right],
$$
we have
$$
f(x)\geq f(x_1)\geq M+(1+Bx_1)^{-s}\geq M+(B/(AM))^{-s/(s+1)}>M+AM/B
$$
for all $x\in [0,1/A)$ and $s>0$. Combining the results in both cases, we obtain the assertion of the lemma. \end {proof}

\medskip

\begin{proof}[{\bf Proof of Theorem \ref {Q_4}.}]
As we mentioned in Proposition~\ref{prop2} above, any cluster point of a family of $s$-energy
minimizing configurations as $s\to\infty$ is a best-packing
configuration. Thus, by Proposition~\ref{l1}, it is sufficient to show
that no 5-point configuration consisting of two opposite poles and an acute triangle on the
equator (which we call an {\em acute configuration}) could be such a cluster point. We will
prove this by contradiction. For $s$ large, consider a
minimal $s$-energy configuration that is `close' to a fixed acute
configuration. We may assume that this minimal $s$-energy
configuration $\omega_5(s)$ consists of three points $$
A_1=A_{1s}=(a_{11s},a_{12s},h), \,
A_2=A_{2s}=(a_{21s},a_{22s},h), \,
A_3=A_{3s}=(a_{31s},a_{32s},h),
$$
where $h=h_s=o(1)$ as $s\to\infty$,
  that are close to the vertices of a fixed
acute triangle on the equator,
and two points $A_{4}=A_{4s}$ and
$A_{5}=A_{5s}$ that are close to $(0,0,1)$ and $(0,0,-1)$,
respectively. Denote by
$$
E_1:=E_{1s}=\sum_{i=1}^3|A_4-A_i|^{-s}, \quad\text{and}\quad
E_2:=E_{2s}=\sum_{i=1}^3|A_5-A_i|^{-s}.
$$
Clearly, the total $s$-energy  $E_s(\omega_5(s))> 2E_1+2E_2$.

Let us first estimate $E_1$ from below. Denote by $O$ the point
$(0,0,h)$,  by $B$  the projection of $A_{4}$ to the plane
$A_{1}A_{2}A_{3}$, and by $x$ the length $|O-B|$. Without lost of
generality we may assume that $B$ lies in the triangle $OA_2A_3$. Here
we use the facts that $x=x_s=o(1)$ as $s\to\infty$, and
that $A_{1}A_{2}A_{3}$ is `close' to a fixed acute triangle implying that $O$ lies inside the triangle $A_{1}A_{2}A_{3}$. Denote by $\alpha=\alpha_s$,
$\beta=\beta_s$, and $\gamma=\gamma_s$ the angles $A_{2}OB$, $A_{3}OB$, $A_{2}OA_{1}$,
respectively (see Figure~\ref{triFig}).

\begin{figure}[htbp] 
    \centering
    \includegraphics[width=3in]{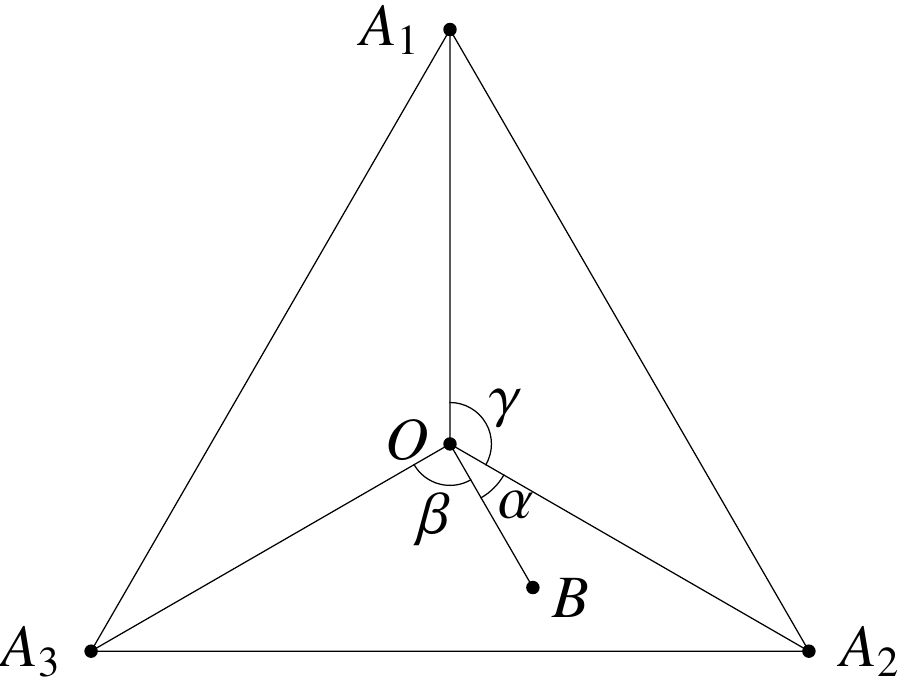}  
    
        \caption{Projection  $B$ of $A_4$ on horizontal $A_1A_2A_3$ plane. }
       \label{triFig}
 \end{figure}

Since 
$$
E_1=\sum_{i=1}^3|A_4-A_i|^{-s}=\sum_{i=1}^3\left(|B-A_4|^2+|B-A_i|^2\right)^{-s/2},
$$
we have, by the law of cosines and the fact that $|B-A_4|=\sqrt{1-x^2}-h$,

\begin{align*}
E_1&=(2-2h\sqrt{1-x^2}-2x\sqrt{1-h^2}\cos \alpha)^{-s/2}\\ &+(2-2h\sqrt{1-x^2}-2x\sqrt{1-h^2}\cos \beta)^{-s/2}\\&+(2-2h\sqrt{1-x^2}-2x\sqrt{1-h^2}\cos (\alpha+\gamma))^{-s/2}.
\end{align*}

 The crucial observation is the fact that $\alpha+\beta<\tau<\pi$, for
some $\tau$ that does not depend on $s$. Now monotonicity and convexity
of the function $t^{-s/2}$, $t>0$, immediately imply
\begin{align}\label{E11}
E_1 & \ge 
2\left(2-2h\sqrt{1-x^2}-x\sqrt{1-h^2}(\cos{\alpha}+\cos{\beta})\right)^{-s/2}\\&\hspace{2in} +(2-2h\sqrt{1-x^2}+2x)^{-s/2}
\nonumber \\
&\ge\nonumber
2\left(2-2h\sqrt{1-x^2}-x\sqrt{1-h^2}(1+\cos{\tau})\right)^{-s/2}\\&\hspace{2in}+(2-2h\sqrt{1-x^2}+2x)^{-s/2}.
\nonumber
\end{align}
From the facts that   $x=o(1)$, and
$h=o(1)$ as $s\to\infty$ and the inequality $1-x\le \sqrt{1-x^2}\le 1$, we get that 
 $$
E_1\ge2\left(2-2h-\theta_1 x\right)^{-s/2} +(2-2h+3x)^{-s/2},
$$
for some absolute constant $\theta_1>0$.  Then, by Lemma~\ref{M},
$$
E_1\ge (2+\theta_2)(2-2h)^{-s/2},
$$
for some absolute constant $\theta_2>0$.   
Similarly we obtain
$$
E_2\ge (2+\theta_2)(2+2h)^{-s/2},
$$
and so  again applying the convexity of $t^{-s/2}$ we finally deduce that, for $s$ sufficiently large,\begin{equation}
\label{min} E_s(\omega_5(s))>2(E_1+E_2)\ge\left(8+4\theta_2\right)2^{-s/2}.
\end{equation}

On the other hand, from Lemma~\ref{S_8}, we know that $\mathcal{E}_s(S^2,5)\le (8+o(1))2^{-s/2}.$
 Therefore, by~\eqref{min}, an
acute configuration cannot be a cluster point of minimal $s$-energy
configurations as $s\to\infty$.\end{proof}

We can now obtain the dominant term in the asymptotic expansion for the minimal 5-point $s$-energy.

\begin{theorem}
\label{last} We have
$$
\lim_{s\to\infty}{2^{s/2}\mathcal E_s(S^2,5)}= 8.
$$
\end{theorem}
\begin{proof}
By Lemma~\ref{S_8} it is enough to prove that
\begin{equation}
\label{in} \liminf_{s\to\infty}{2^{s/2}\mathcal
E_s(S^2,5)}\ge 8.
\end{equation}
For a fixed $s>0$ consider a minimal $s$-energy configuration $\omega_5(s)=\left\{
A_1,\ldots, A_5\right\}:=\left\{ A_{1s},\ldots, A_{5s}\right\}$. By
Theorem~\ref{Q_4} we may assume that both distances $|A_2-A_3|$ and
$|A_4-A_5|$ have limit $2$ as $s\to\infty$. Observe that if the triangle
$A_1A_2A_3$ is not acute, then
$$
|A_1-A_2|^{-s}+|A_1-A_3|^{-s}\ge |A'_1-A_2|^{-s}+|A'_1-A_3|^{-s}\ge 2^{1-s/2},
$$
where $A'_1$ is the midpoint of the circular arc (of length less than $\pi$) joining $A_2$ and $A_3$ and containing $A_1$.  
A similar statement holds for triangle $A_1A_4A_5$. Therefore  we may assume that at least one of the triangles
$A_1A_2A_3$ or $A_1A_4A_5$ is acute since otherwise the desired lower bound for the $s$-energy follows. Without lost of generality, we
assume that $A_1A_2A_3$ is acute.  We adopt the same notation as in the proof of Theorem~\ref{Q_4} and obtain a finer lower bound for $E_1$ and $E_2$.   

 There are three possible
cases to consider, depending on the location of the projection $B$ of $A_4$ onto the plane containing $A_1$, $ A_2$, and $A_3$:  
(i) $B$ is inside the sector $A_2OA_3$; (ii)
$B$ is inside the sector $A_1OA_2$; and (iii) $B$ is inside the sector $A_1OA_3$.

Let us assume first that $B$ is inside the sector $A_2OA_3$ as in Figure~\ref{triFig}.
From \eqref{E11}, we get 
$$
E_1\ge 2
(2-2h\sqrt{1-x^2}-x\sqrt{1-h^2}(\cos\alpha+\cos\beta))^{-s/2}\ge
2(2-2h\sqrt{1-x^2})^{-s/2}.
$$
In both other cases (ii) and (iii) we get the same inequality. Letting $D$ denote the
projection of $A_5$ onto the plane $A_1A_2A_3$ and setting $y=|O-D|$, we similarly
get  
$$
E_2\ge 2(2+2h\sqrt{1-y^2})^{-s/2}.
$$
Thus,
$$
E_s(\omega_5(s))>4\left[(2-2h\sqrt{1-x^2})^{-s/2}+(2+2h\sqrt{1-y^2})^{-s/2}\right].
$$
Finally applying Lemma~\ref{M} to the last inequality and using the
fact that $x=o(1)$ and $y=o(1)$ as $s\to\infty$ we immediately
obtain~\eqref{in}.
\end{proof}

\section{Special best-packing configurations on $S^n$}

 In the case   $A=S^n$ with $n\ge 2$ and 
Euclidean distance, there are best-packing configurations $\omega_N$
 such that $\eta(\omega_N,S^n)=\delta(\omega_N)=\sqrt{2}$  for
$N=n+3,\ldots, 2n+1$, yielding $\gamma(\omega_N,S^n)=1$ (see Theorem 6.2.1~\cite{KB}).
For $N=5$ on $S^2$,  such a configuration is given by SBP$(\infty)$ defined in \eqref{Qp}.

 By the proof of
Theorem~\ref{Thm0}, we have $\eta(\omega_N,A)\ge\delta(\omega_N)$
for some best-packing configuration $\omega_N$ if and only if
$\delta_N(A)=\delta_{N+1}(A)$, which should be a very rare
event, at least for $A=S^n$ .    For $S^2$ and $N=11$ there exists a unique (up to isometry) best-packing
configuration $\omega_{11}$ consisting of the regular icosahedron minus one of its vertices (see \cite{B2}). Hence, \begin{equation}
\label{11p} \eta_{11}(S^2)=\delta_{11}(S^2) \text{ and } \gamma(\omega_{11},S^2)=1.
\end{equation}

The unique best-packing configuration of $120$ points on $S^3$ is
the $600$-cell  configuration which has many other
fascinating extremal properties, see~\cite{A,CK}. Moreover,
in~\cite{NS}, the numerical evidence is given that
\begin{equation}
\label{113p}
\delta_{113}(S^3)=\ldots=\delta_{120}(S^3)=(\sqrt{5}-1)/2.
\end{equation}
Assuming~\eqref{113p}, we are able to construct a best-packing
configuration of $113$ points on the sphere with
$\eta(\omega_{113},S^3)>\delta_N(\omega_{113})$. It consists of
$600$-cell without certain 7 points which we describe below.

In the $600$-cell each point has 12 other points at the closest
distance $(\sqrt{5}-1)/2$, and each pair of points at this distance
has exactly 5 other points having the same distance to both
points of the pair. So we will remove two points $x_1, x_2$, such
that $
|x_1-x_2|=(\sqrt{5}-1)/2$, and also 5 points $y_1,\ldots
y_5$, such that $|x_i-y_j|=(\sqrt{5}-1)/2$, $i=1,2$,
$j=1,\ldots,5$. Recall that the second largest distance between points
of the $600$-cell is~1. Thus,
$$
\eta(\omega_{113},S^3)\ge\min_{x\in\omega_{113}}\left|\frac{x_1+x_2}{|x_1+x_2|}-x\right|=
\sqrt{2-\frac{3+\sqrt{5}}{\sqrt{10+2\sqrt{5}}}}\approx
1.2778\,\delta(\omega_{113}).
$$

\noindent
{\bf Acknowledgements.}  The authors thank the Mathematisches Forschungsinstitut Oberwolfach for their hospitality during the preparation of this manuscript and 
for providing a stimulating atmosphere for research.

\end{document}